\documentclass[]{article}  
\usepackage{url,float}
\usepackage{graphicx}
\usepackage{amsmath}
\usepackage{amsfonts}
\usepackage{amssymb}
\usepackage{latexsym}
\usepackage{mathtools}
\usepackage{color}

\newcommand{\hide}[1]{}

\newcommand{\ABox}{
\raisebox{3pt}{\framebox[6pt]{\rule{6pt}{0pt}}}}
\newenvironment{proof}{{\bf Proof:}}{\hfill\ABox}
\newtheorem{theorem}{{\bf Theorem}}

\newtheorem{lemma}{Lemma}

\newtheorem{definition}[theorem]{Definition}

\newcommand{\lemlab}[1]{\label{lemma:#1}}
\newcommand{\thmlab}[1]{\label{thm:#1}}

\newcommand{\figlab}[1]{\label{fig:#1}}
\newcommand{\seclab}[1]{\label{sec:#1}}

\newcommand{\figref}[1]{\ref{fig:#1}}

\def\a{{\alpha}}
\def\P{{\mathcal P}}
\def\C{{\mathcal C}}
\def\V{{\mathcal V}}
\def\S{{\Sigma}}
\def\o{{\omega}}

\def\bC{{\partial C}}
\def\R{{\mathbb{R}}}

\newcommand{\squeezelist}{\setlength{\itemsep}{0pt}}

\title{%
A Note on Unbounded Polyhedra\\
Derived from Convex Caps
} 

\author{%
Joseph O'Rourke
}
\begin{document}
\maketitle
\begin{abstract}
The construction of an unbounded polyhedron from a
``jagged'' convex cap is described, and several of its
properties discussed, including its relation to
Alexandrov's ``limit angle."
\end{abstract}

\section{Introduction}
\seclab{Introduction}
One step of the proof in~\cite{o-eunfcc-17}
extended a convex polyhedral cap 
$\C$ an unbounded polyhedron $\C^\infty$.
In this note, we explore some properties of such unbounded polyhedra,
including 
Alexandrov's ``limit angle''~\cite[p.29ff]{a-cp-05}.
In particular, we show:
\begin{enumerate}
\squeezelist
\item A ``jagged" convex cap is homeomorphic to a disk.
\item How to construct the unbounded polyhedron $\C^\infty$ from a given cap $\C$.
\item How to construct Alexandrov's limit angle $\V$.
\item The relationship between the curvature of the limit angle apex to the curvature
at the vertices of $\C$ and of $\C^\infty$.
\end{enumerate}
None of these results are 
new and none surprising,
so these remarks amount to a tutorial on the topic.

\section{Convex Cap}
\seclab{ConvexCap}
In~\cite{o-eunfcc-17} a convex cap was defined as the intersection of
a half-space with a convex polyhedron:
\begin{quotation}
\noindent
``Let $\P$ be a convex polyhedron, and let $\phi(f)$  be
the angle the normal to face $f$ makes with the $z$-axis.
Let $H$ be a halfspace whose bounding plane is orthogonal to the $z$-axis, and includes points
vertically above that plane.
Define a \emph{convex cap} $\C$ of angle $\Phi$ to be $C=\P \cap H$
for some $\P$ and $H$, such that $\phi(f) \le \Phi$ for all $f$ in $\C$.\footnote{
This definition accords with Alexandrov's ``polyhedral cap"\cite[p.184]{a-cp-05}.}
\end{quotation}
Here I'd like to loosen this definition to allow a jagged boundary
rather than the planar boundary obtained by intersection with $H$:
\begin{definition}
A \emph{jagged convex cap} $\C$ of angle $\phi$ is the collection of
all faces of a convex polyhedron $\P$ each of whose face-normals makes an angle 
strictly less than $\phi$ to the (vertical) $z$-axis.
\end{definition}
As in~\cite{o-eunfcc-17},
we only consider $\phi \le 90^\circ$, which implies that the projection
of $\C$ onto the $xy$-plane is one-to-one.
(If  $\phi=90^\circ$, face normals of $\C$ make an angle $< 90^\circ$,
so no face of $\C$ is vertical.)
Note that $\C$ is not a closed polyhedron; it has no ``bottom,''
but rather a ``jagged" boundary $\bC$.
An example is shown in 
Fig.~\figref{Polyh_s1_n50}.
We will henceforth abbreviate ``jagged convex cap" to ``cap."
\begin{figure}[htbp]
\centering
\includegraphics[width=0.75\linewidth]{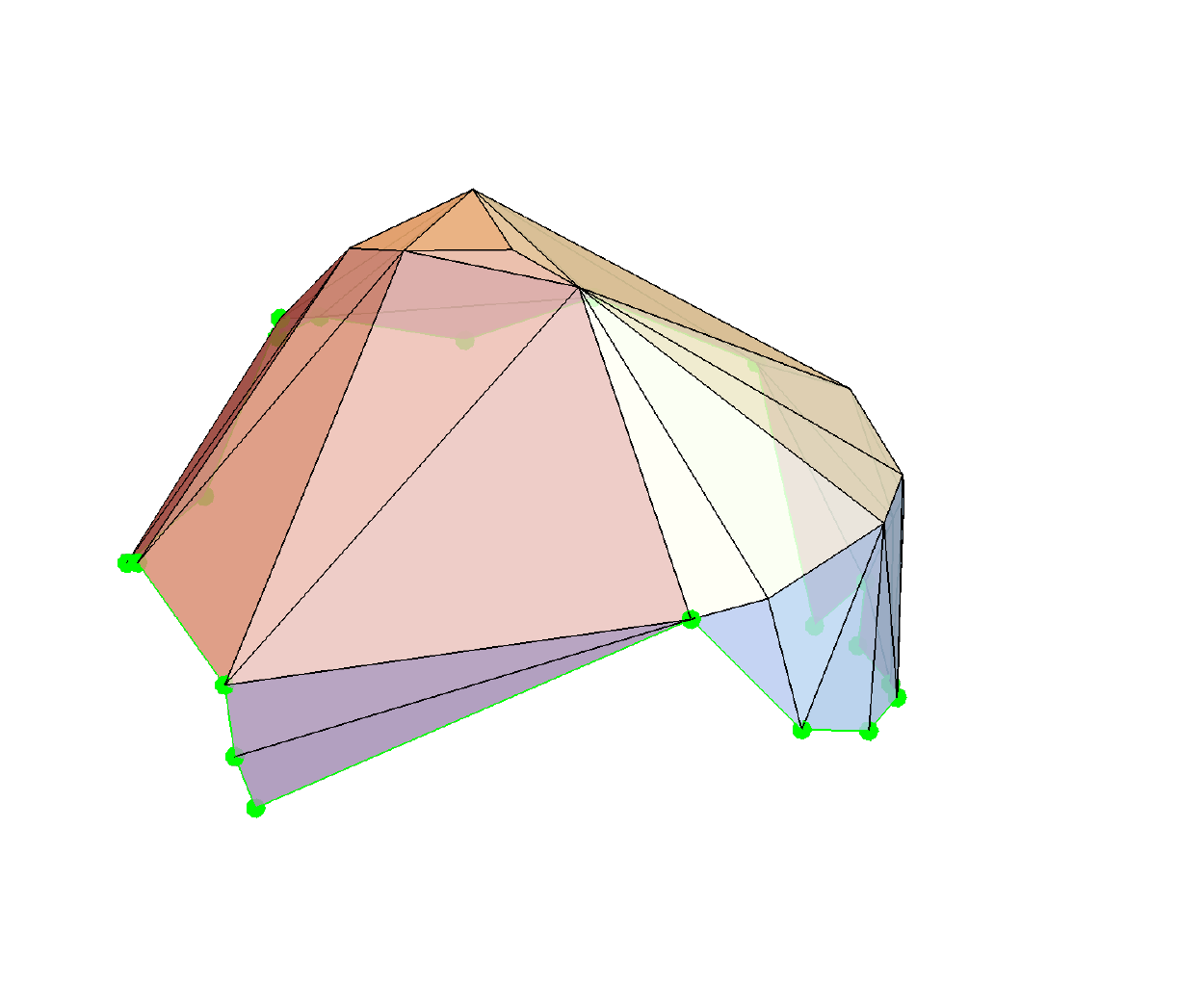}
\caption{A convex cap $\C$ of $31$ vertices, $\phi=90^\circ$.
The boundary $\bC$ is marked in green.}
\figlab{Polyh_s1_n50}
\end{figure}

If $\C$ results from the intersection of $\P$ with a half-space, then it is
evident it is topologically (homeomorphic to) a disk. The same holds for the new definition, but perhaps not as evidently:\footnote{
This seems assumed in various convex hull algorithms, but I could not find
a proof.}
\begin{lemma}
A jagged cap $\C$ (as defined above) is topologically a disk,
and so its boundary $\bC$ is topologically a circle,
i.e., $\bC$ is a simple polygonal cycle.
\lemlab{disk}
\end{lemma}
\begin{proof}
It seems easiest to see this for $\phi=90^\circ$, because then
this becomes a version of the ``shadow problem''~\cite{g-ssp3s-2002}.\footnote{
I have only found literature for smooth convex bodies, as opposed to
polyhedra.}
For $\phi=90^\circ$, we can view the cap $\C$ as the portion of the polyhedron
$\P$ illuminated by a light at $z=+\infty$. The remainder of $\P$ is in shadow,
and the shadow boundary $\bC$ is the collection of shadow edges separating
light above from dark below.
Note that faces of $\P$ that are themselves vertical are considered in
shadow, because only those faces whose normals make an angle strictly
less than $\phi$ belong to the cap and so are illuminated.

Let $\Pi$ be a vertical plane supporting $\P$.
\begin{itemize}
\squeezelist
\item If $\Pi$ contains a single vertex $v$, then $v$ is part of the shadow boundary:
$v \in \bC$.
\item If $\Pi$ contains a single edge $e$, then $e \in \bC$.
\item If $\Pi$ contains three or more noncollinear points of $\P$, then
$\Pi$ contains a face of $\P$. Because that face is not illuminated, only
its upper edges are part the shadow boundary $\bC$.
See Fig.~\figref{VertPlanes}(a).
\end{itemize}
\begin{figure}[htbp]
\centering
\includegraphics[width=0.75\linewidth]{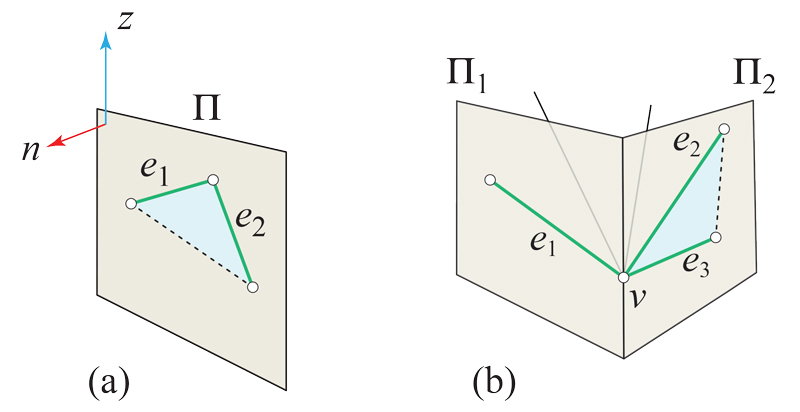}
\caption{(a)~$e_1,e_2 \in \bC$; the face in $\Pi$ is not illuminated. 
(b)~A vertex of degree-$3$ determines a vertical face.
$e_3 \not\in \bC$.}
\figlab{VertPlanes}
\end{figure}

In order to prove the shadow boundary is a simple polygonal cycle,
consider a vertex $v \in \bC$ in a vertical plane $\Pi$.
Rotate $\Pi$ about a vertical line through $v$ until it hits $\P$ to
the left, and rotate in the other direction until it hits $\P$ to the right.
If either of these planes $\Pi_1$ or $\Pi_2$ contains more than one edge incident to $v$
(as does $\Pi_2$ in Fig.~\figref{VertPlanes}(b)), then again we have
identified a face in the vertical plane, and only the upper edges of that
face are part of $\bC$ ($e_2$ in the figure).
Thus each vertex $v \in \bC$ is incident to exactly two edges in $\bC$.
So these edges form a simple polygonal cycle.

It only remains to show that the shadow boundary is connected,
i.e., there is not more than one such cycle.
If $\bC$ has a disconnected hole, then a vertical plane through
an upper edge of the hole would also have to include (a portion of)
a lower edge, requiring the hole to determine a vertical face, which
we've seen is not part of $\bC$. 

Now consider $\phi < 90^\circ$.
For any vector $u$ along a generator\footnote{
A \emph{generator} of a cone is a line through the cone apex and lying in the
cone surface.}
of the $\phi$-cone surrounding $z$, let
$\Pi$ be a plane orthogonal to $u$.
Sweep $\Pi$ inward/downward until it hits $\P$; see Fig.~\figref{Phi60}.
The upper edges of these contacts, for all $u$, constitute
the equivalent of the shadow boundary. Again $\Pi$ might
be flush with a face of $\P$, but then that face is not illuminated
because its normal is exactly $\phi$ from $z$.
So only its upper edges are part of $\bC$. 
\begin{figure}[htbp]
\centering
\includegraphics[width=0.75\linewidth]{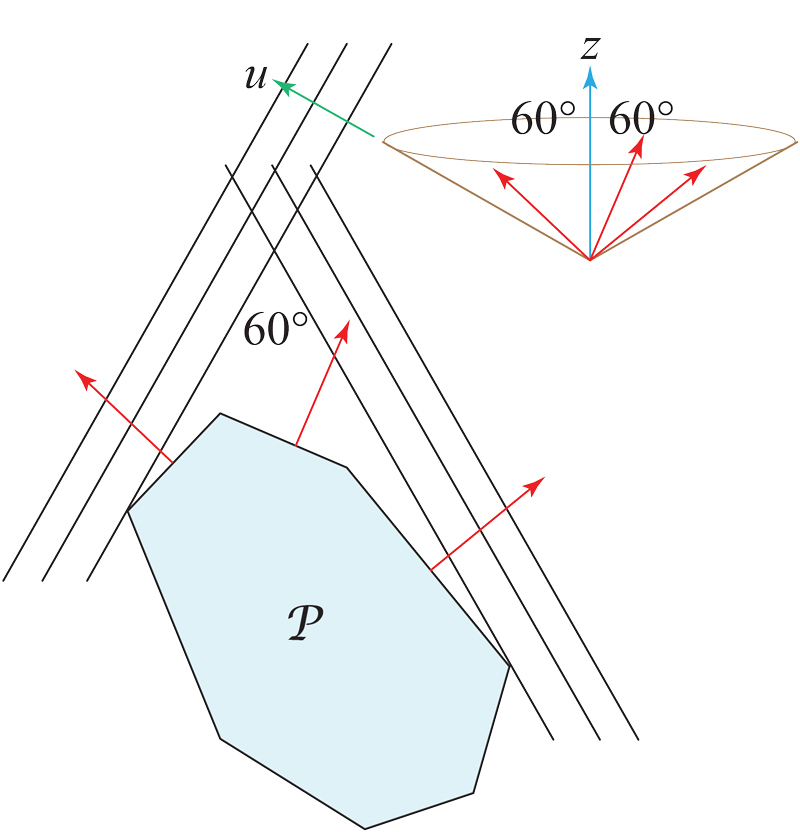}
\caption{Side-view depiction for $\phi=60^\circ$.}
\figlab{Phi60}
\end{figure}
The reasoning used for $\phi = 90^\circ$ now goes through unchanged.
\end{proof}

\section{Extension of Cap to Unbounded Polyhedron}
\seclab{Extension}
Any convex cap $\C$ can be extended to an unbounded polyhedron 
$\C^\infty$.
We first illustrate this before describing a procedure to construct $\C^\infty$
from $\C$.
The extension may be achieved by augmenting $\C$ with the
intersection of half-spaces bounded by planes containing
all the \emph{boundary faces} of $\C$, that is, the faces that share
an edge with $\bC$.
Note that there is no need to include faces only incident to a vertex
of $\bC$, as their extensions fall outside $\C^\infty$.

\begin{figure}[htbp]
\centering
\includegraphics[width=1.0\linewidth]{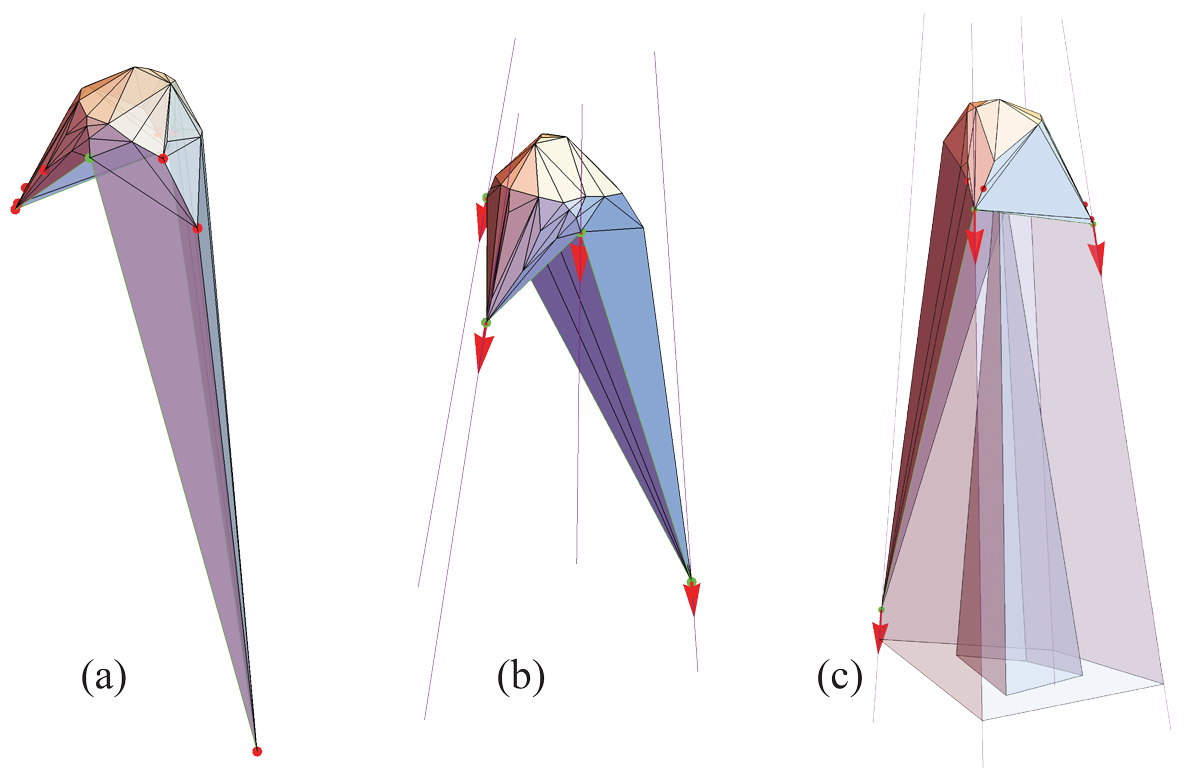}
\caption{(a)~Extension of $\C$ in Fig.~\protect\figref{Polyh_s1_n50}.
Vertices created by the extension marked red.
(b)~Unbounded edges.
(c)~Unbounded faces and the limit angle.
 (Images not to same scale, nor viewed from the same angle.)}
\figlab{Triple_s1_n50}
\end{figure}

In general, $\C^\infty$ includes vertices not in $\C$, and three
or more unbounded edges (rays) bounding unbounded faces.
Fig.~\figref{Triple_s1_n50}(a)
shows all but the unbounded faces of $\C^\infty$ for the cap $\C$ shown
earlier in Fig.~\figref{Polyh_s1_n50}.
Note that the number of unbounded faces of $\C^\infty$ is generally
smaller than the number of boundary faces of $\C$.
The unbounded edges do not (in general) meet in a point if extended upward,
but if those rays are brought together, they delimit
Alexandrov's ``limit angle," which is shown in 
Fig.~\figref{Triple_s1_n50}(c).

We now describe one route to calculate $\C^\infty$.
\section{Constructing the Extension of a Cap}
\seclab{Dual}
Because of the wide availability of code to compute hulls from points
in $\R^3$ compared to the apparent paucity of code to
directly intersect half-spaces, we use the well-known duality
to intersect the half-spaces via a hull of points in a dual space.\footnote{
E.g.,\cite[Ch.25]{hp-gaa-2011}.}
In particular, these are the steps followed:

\begin{enumerate}
\squeezelist
\item For every boundary face of $\C$, compute the plane
equation determined by the face:
$z=ax+by+c$.
\item Dualize each plane to the point $(a,b,-c)$. Call the set of these point $P^*$.
\item Take the convex hull of $P^*$. The upward-faces of this hull
form the lower envelope $L$, and
represents the intersection of the half-spaces.
\item Determine the plane equation for each upward-face of $L$ and dualize back, forming
the set of points $P^{**}$. These are the points of intersection of three
planes of the extension. In Fig.~\figref{Triple_s1_n50}(a), $10$ vertices are added.
\item Join $P^{**}$ with the vertices of $\C$, and again take the convex hull.
Retain only the upward-facing faces. This represents all of the bounded
part of $\C^\infty$
missing only the unbounded faces and edges. Call this $\C'^\infty$.
\item Construct the unbounded edges from the boundary faces
of $\C'^\infty$
(the red rays in Fig.~\figref{Triple_s1_n50}(b)), and from those the unbounded faces. 
We know all the boundary faces of $\C'^\infty$ extend to unbounded faces,
with each adjacent pair of these faces intersecting in an unbounded edge, so this computation is straightforward.
Now $\C^\infty$ is complete.
\end{enumerate}

We illustrate with another similar example before turning to
Alexandrov's limit angle.
See Figs.~\figref{Top_s5_n50} and~~\figref{Triple_s5_n50}.
\begin{figure}[htbp]
\centering
\includegraphics[width=0.85\linewidth]{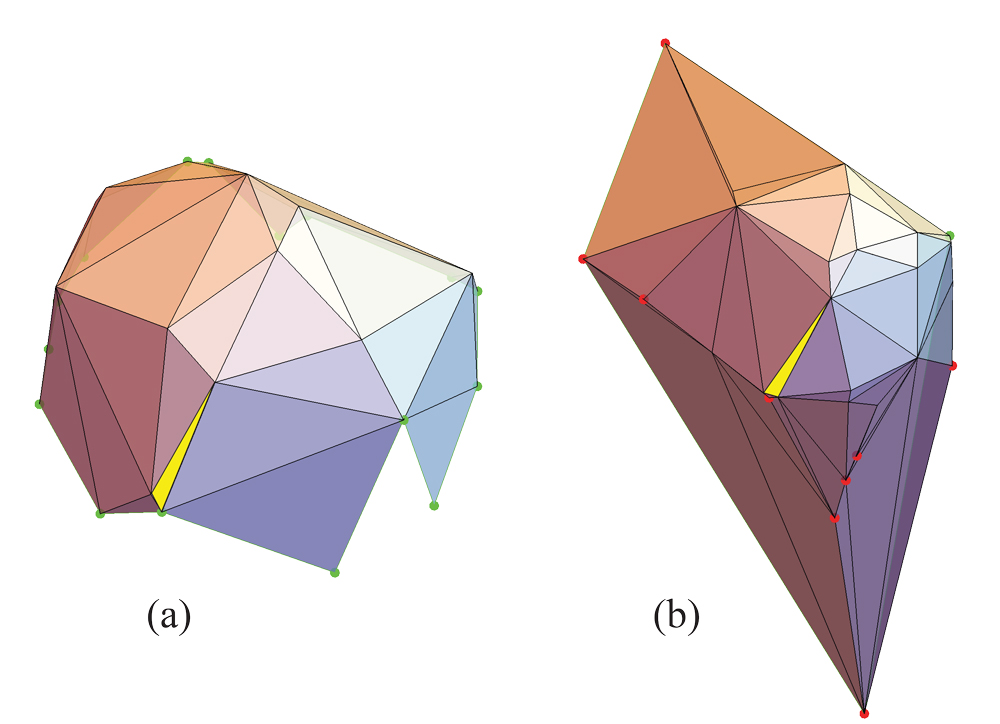}
\caption{(a)~A cap $\C$ of $27$ vertices.
(b)~$\C'^\infty$, the bounded portion of $\C^\infty$, top view (greatly distorted). 
One face
is marked (yellow) for orientation. New vertices are red.}
\figlab{Top_s5_n50}
\end{figure}

\begin{figure}[htbp]
\centering
\includegraphics[width=0.9\linewidth]{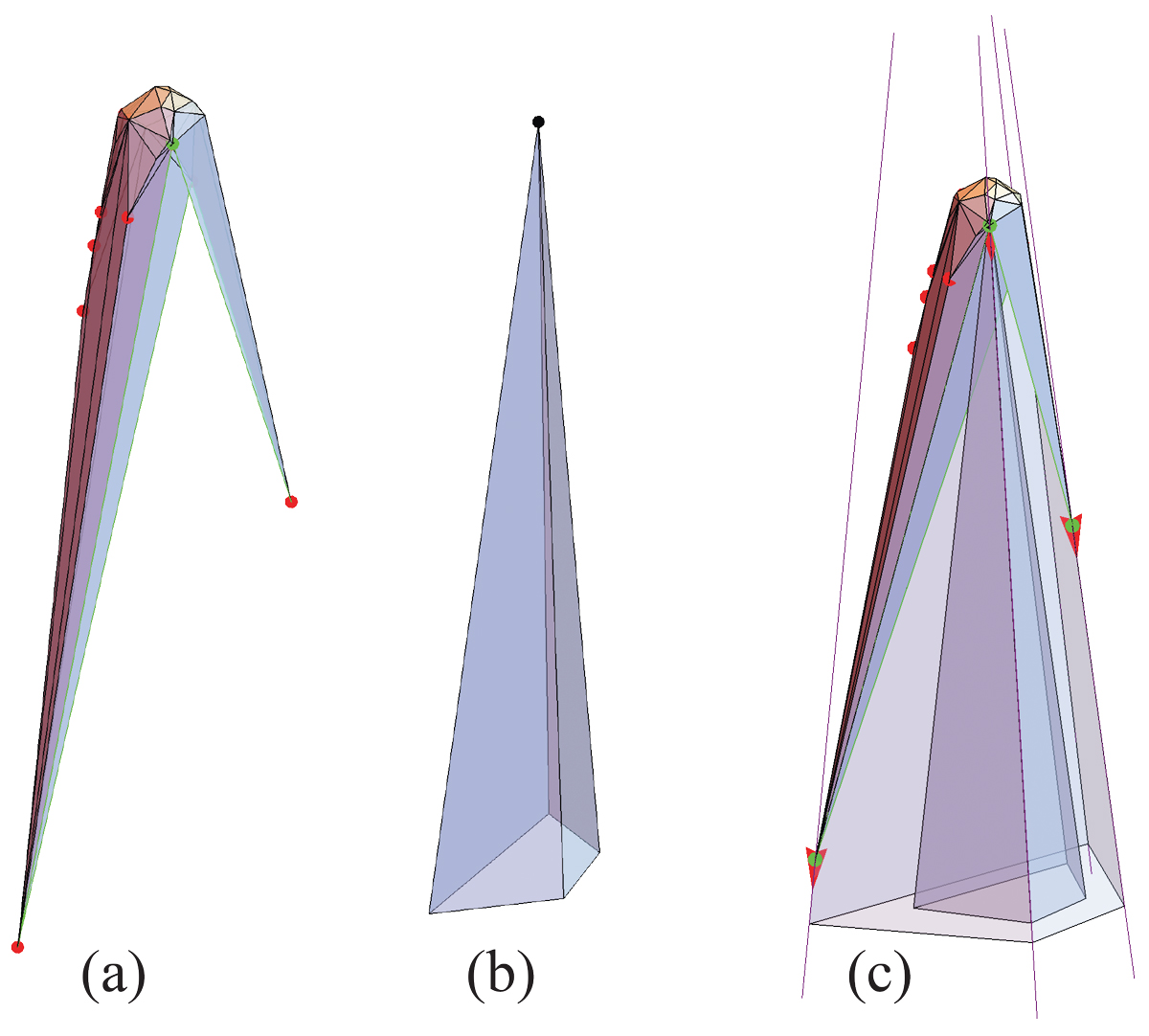}
\caption{(a)~Extension of cap in Fig.~\protect\figref{Top_s5_n50}(a),
side view.
$9$ new vertices (red) added to $\C$ by extending faces.
(b)~Limit angle $\V$.
(c)~$\C^\infty$ and limit angle inside. 
(Images not to same scale, nor viewed from the same angle.)}
\figlab{Triple_s5_n50}
\end{figure}

\newpage
\section{Alexandrov's Limit Angle}
\seclab{AlexLA}

Alexandrov~\cite[p.29]{a-cp-05} proves this theorem
defining the \emph{limit angle}:

\begin{theorem}[\cite{a-cp-05}]
Let $\P$ be an unbounded convex polyhedron.
If we draw all rays from some point $O$ that are parallel to the rays
lying in $\P$, then we obtain the \emph{limit angle} $\V$ of $\P$.
This is a convex polyhedral angle whose edges and faces are parallel
to the unbounded edges and faces of the polyhedron.
If the point $O$ lies in $\P$, then the angle $\V$ lies in $\P$.
\thmlab{LA}
\end{theorem}

After calculating the unbounded edges of  $\C^\infty$, it is easy to
construct $\V$, which is displayed inside  $\C^\infty$
in Fig.~\figref{Triple_s1_n50}(c), and separately in~\figref{Triple_s5_n50}(b).
The construction simply joins to a point $O$ 
vectors along the unbounded edges.

\subsection{Curvature of the Limit Angle}
\seclab{Curvature}
The curvature at a polyhedron vertex $v_i$ is $\o(v_i) = 2 \pi - \S \a_j$,
where $\a_j$ are the face angles incident to $v_i$.
The curvature of the limit angle of an unbounded polyhedron
such as $\C^\infty$ is equal to the sum of the curvatures at
all vertices of  $\C^\infty$: $\S \o(v_i)$ over all $v_i \in C^\infty$.
For a closed polyhedron, this sum is of course $4 \pi$ by
Gauss-Bonnet, but for $\C^\infty$, with $\phi \le 90^\circ$, 
the sum is always strictly less than $2 \pi$.
Note that the curvature of the limit angle is not the sum
of the curvatures of the vertices of the cap $\C$: The new vertices
created by extending to $\C^\infty$ each  have some positive curvature
that contribute to the limit angle total.

Alexandrov proved the identity between curvatures via the \emph{spherical image},
(also known as the ``Gaussian sphere")~\cite[p.43]{a-cp-05}: 

\begin{theorem}[\cite{a-cp-05}]
The spherical image of an unbounded convex polyhedron
coincides with the spherical image of its limit angle.
\thmlab{Spherical}
\end{theorem}
This suffices to prove equality, because the area of the spherical image of a (convex) vertex $v_i$ is precisely
that vertex's curvature $\o(v_i)$.

Two further remarks on the limit angle curvature:
\begin{enumerate}
\squeezelist
\item Alexandrov imagines an ``infinite similarity contraction''~\cite[p.28]{a-cp-05}
which in the limit reduces an unbounded polyhedron to the limit angle.
During the contraction, the curvature of all vertices remains fixed, because
all edges remain parallel to themselves (so all face angles remain fixed).
So it is natural that in the limit, all the curvature is transferred to the limit angle
vertex.
\item In~\cite{ov-dcpce-14} we detailed a method to merge two vertices into one,
retaining the overall curvature of the cap. Repeating this reduces any unbounded polyhedron
to a cone whose one vertex has the same curvature
$\S \o(v_i)$ as the limit angle curvature.
\end{enumerate}

\section{Summary}
\seclab{Summary}
\begin{enumerate}
\squeezelist
\item A jagged convex cap is a topological disk.
\item The extension of a cap $\C$ to an unbounded polyhedron $\C^\infty$
can be constructed using plane-point duality. In general $\C^\infty$ has
more vertices than $\C$.
\item The unbounded edges of $\C^\infty$ do not (generally) meet in a point
when extended upward, but if translated to a common point $O$,
they form Alexandrov's limit angle $\V$.
\item The curvature at the apex of the limit angle $\V$ is the same as the 
sum of the curvatures of all vertices of $\C^\infty$.
\end{enumerate}

\bibliographystyle{alpha}
\bibliography{Unbounded}
\end{document}